\documentclass[a4paper,english,numberwithinsect]{lipics-v2018}

\usepackage[T1]{fontenc}
\usepackage{latexsym,color,amsmath,amssymb,amsthm}
\usepackage{xcolor}
\usepackage{graphicx}
\usepackage{microtype}
\usepackage{xspace}

\graphicspath{{figs/}}

\newcommand{\reals}{\mathbb R}
\newcommand{\eps}{\varepsilon}
\def\D{\EuScript{D}}
\newcommand{\M}{\EuScript{M}}
\def\X{\mathsf{X}}
\DeclareMathOperator{\VD}{VD}
\DeclareMathOperator{\VC}{VC}
\DeclareMathOperator{\polylog}{polylog}

\DeclareMathOperator{\cost}{cost}
\DeclareMathOperator{\optcost}{cost\lower 0,3ex
  \hbox{$\mskip -1,5mu^*\mskip-1,5mu$}} 
\newcommand{\softO}{\widetilde{O}}
\newcommand\match{M}
\newcommand\approxmatch{\widetilde{M}}
\def\approxmap{\widetilde{\M}}

\newcommand{\etal}{\emph{et~al.}\xspace}

\def\marrow{\marginpar[\hfill$\longrightarrow$]
     {\textcolor{red}{$\longleftarrow$}}}
\newcommand{\remarkB}[3]{\marrow\textcolor{blue}{\textsc{#1 #2:}} 
     \textcolor{red}{\textsf{#3}}}
\renewcommand{\remarkB}[3]{}

\title{Approximate Minimum-Weight Matching with Outliers under
  Translation 
}

\titlerunning{Approximate Minimum-Weight Partial Matching under 
Translation}

\author{Pankaj K. Agarwal}{Department of Computer Science, 
Duke University, Durham, NC 27708, USA}{pankaj@cs.duke.edu}{}
{Partially supported by NSF grants 
CCF-15-13816, CCF-15-46392, IIS-14-08846 and ARO 
grant W911NF-15-1-0408.}
\author{Haim Kaplan}{School of Computer Science, Tel Aviv University, 
Tel~Aviv 69978, Israel}{haimk@tau.ac.il}{}{}
\author{Geva Kipper}{School of Computer Science, Tel Aviv University, 
Tel~Aviv 69978, Israel}{gevakip@gmail.com}{}{}
\author{Wolfgang Mulzer}{Institut f\"ur Informatik, 
Freie Universit\"at Berlin, 14195 Berlin, Germany}
{mulzer@inf.fu-berlin.de}
{https://orcid.org/0000-0002-1948-5840}
{Partially supported by DFG grant MU/3501/1 and ERC STG 757609.}
\author{G\"unter Rote}{Institut f\"ur Informatik, 
Freie Universit\"at Berlin, 14195 Berlin, Germany}
{rote@inf.fu-berlin.de}{https://orcid.org/0000-0002-0351-5945}{}
\author{Micha Sharir}{School of Computer Science, Tel Aviv University, 
Tel~Aviv 69978, Israel}{michas@tau.ac.il}{}{Partially supported by 
ISF Grant 892/13, by the Israeli Centers of Research 
Excellence (I-CORE) program (Center No.~4/11), by the 
Blavatnik Research Fund in Computer Science at Tel Aviv University,
and by the Hermann Minkowski-MINERVA Center for Geometry at 
Tel Aviv University.}
\author{Allen Xiao}{Department of Computer Science, Duke University, 
Durham, NC 27708, USA}{axiao@cs.duke.edu}{}{Partially supported by 
NSF CCF-15-13816, CCF-15-46392, IIS-14-08846 and ARO 
grant W911NF-15-1-0408.}

\authorrunning{P. K. Agarwal, H. Kaplan, G. Kipper, W. Mulzer, 
G. Rote, M. Sharir, A. Xiao}

\Copyright{Pankaj K. Agarwal, Haim Kaplan, Geva Kipper, Wolfgang Mulzer, 
G\"unter Rote, Micha Sharir, Allen Xiao}



\funding{Work on this paper was supported by grant 2012/229 from the 
U.S.--Israel Binational Science Foundation and by grant 
1367/2016 from the German-Israeli Science Foundation (GIF).}

\acknowledgements{This work was initiated at the FUB-TAU 
Joint Research Workshop 
on \emph{Algorithms in Geometric Graphs}, held 
24--28 September 2017 at the Institut f\"ur Informatik
at Freie Universit\"at Berlin. We would like to thank all the
participants of the workshop for creating a conducive research
atmosphere.}

\nolinenumbers 
\hideLIPIcs  

\begin{document}

\maketitle

\begin{abstract}
Our goal is to compare two planar point sets by finding subsets 
of a given size such that a minimum-weight matching 
between them has the smallest weight. This can be done 
by a translation of one set that minimizes 
the weight of the matching. We give efficient algorithms 
(a)  for finding approximately optimal matchings, when the cost 
of a matching is the $L_p$-norm of the tuple of the Euclidean 
distances between the pairs of matched points, for any 
$p\in [1,\infty]$, and (b)~for constructing small-size approximate 
minimization (or matching) diagrams: partitions of the translation space 
into regions, together with an approximate 
optimal matching for each region.
\end{abstract}

\section{Introduction}

The following problem arises in pattern matching:
given point sets $A$, $B$, with
$|A| = m$
and
$|B| = n$,
and $k\le\min\{m,n\}$, find 
subsets $A'\subseteq A$ 
and $B'\subseteq B$ with $|A'| = |B'| = k$ and a transformation $R$ 
that matches $R(A)$ and $B$ as closely as possible,
see Figure~\ref{fig:example}.
\begin{figure}[ht]
  \centering
  \includegraphics[scale=1]{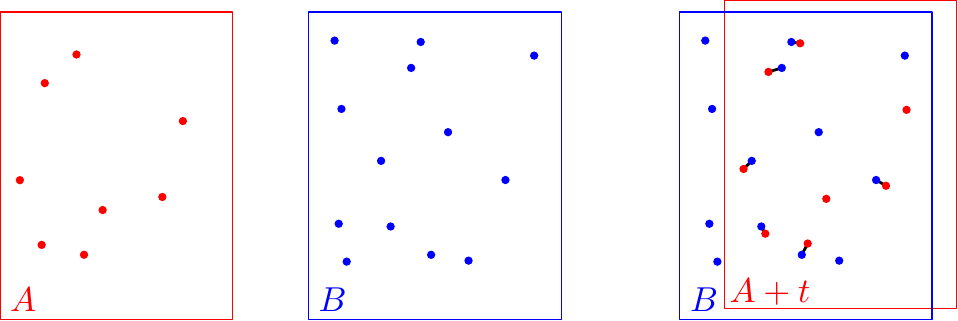}
  \caption{Two sets $A$ and $B$, and a matching of size $k=6$ after
    translation}
  \label{fig:example}
\end{figure}
We think of $A$ as a collection of features, or interest points of
some pattern, that we want to match, bijectively, with
similar features in a large image $B$. Moreover, since the 
coordinate frames 
for $A$ and $B$ are not necessarily aligned, we want to transform
 $A$ to get the best possible fit.

This problem comes in many variants, depending on the class of 
permissible transformations $R$ and on the similarity measure for the 
match.  Here, we want to match $A'$ and $B'$ in a
one-to-one manner, where the cost of a matching depends on 
the distances between matched points.  
Moreover, we only consider translations as permissible 
transformations, and write $A+t$ for the set $A$ translated 
by a vector $t\in\reals^2$.
A
feasible solution is given by a translation $t \in \reals^2$ and by a
\emph{matching} $\match \subset A \times B$ of size $k$ (in short, a \emph{$k$-matching}): a set of $k$ 
pairs $(a,b) \in  A \times B$ so that any point $a \in A$ or $b\in B$
occurs in at most one pair. The parameter $k$ is part of the input.
We consider the $L_p$-cost of such a solution, for 
some~$p \in  [1, \infty]$:
\begin{equation}
  \label{eq:cost}
  \cost_p(\match, t)=
\cost(\match, t)
  := \begin{cases}
\left [ \frac{1}{k} \sum_{(a,b) \in \match} \|a + t - b\|^p \right
]^{1/p}&
\text{ for finite $p$,}\\
\max_{(a,b) \in \match} \|a + t-b\|
&\text{ for $p=\infty$.}\\
\end{cases}
\end{equation}
We will regard $p$ as a fixed constant and will omit it from the notation.
Noteworthy special cases arise when $p = 1$ (sum of distances, 
minimum-weight Euclidean matching), $p = 2$ (root-mean-square matching, 
in short RMS matching), and $p = \infty$ (bottleneck matching). 
In (\ref{eq:cost}), 
we always measure the distances $\|a + t-b\|$ by the Euclidean norm. 
It is not hard to extend the treatment to
other norms, but we stick with Euclidean distances for simplicity.

One important special case occurs when we have a small point set $A$ 
(the \emph{pattern}) that we want to locate within a larger set $B$ 
(the \emph{image}), and $k = |A| < |B|$.  This problem was considered 
for $p=2$ by Rote~\cite{Ro10} and in subsequent 
work~\cite{AHJ*,HJK}, under the name \emph{RMS partial matching}.
Another important instance has $|A| \approx |B|$ and $k$ 
slightly smaller than $|A|, |B|$. Now, we
want to  discard a few \emph{outliers} from each set, to
allow for some 
erroneous data.

For a fixed translation vector 
$t \in \reals^2$, we define 
$ 
  \optcost(t) = \min_{\match} \cost(\match, t) 
$
to be the cost of the minimum-cost $k$ matching 
between $A + t$ 
and $B$. We set 
$
  \match_t = \arg\min_{\match} \cost(\match, t)
$
to be an optimal matching from $A + t$ to $B$, i.e., 
$\optcost(t) = \cost(\match_t, t)$.

Let $\Pi$ be the set of all $k$-matchings from $A$ into $B$. 
The function $\optcost$ is the \emph{lower envelope} (i.e., the
pointwise minimum) of the set of 
functions $F = \{t \mapsto \cost(\match, t) \mid \match \in \Pi\}$. 
The vertical projection of this lower envelope induces a 
planar subdivision, called the \emph{minimization diagram} 
of $F$. It is denoted by $\M := \M(A, B)$. Each face $\sigma$ of
$\M$ is a maximal connected set of points $t$ for which $\optcost(t)$ is 
realized by the same matching $\match_\sigma$. The \emph{combinatorial 
complexity} of $\M$ is the number of its faces. We refer to $\M$ as the
\emph{($k$-)matching diagram} of $A$ and $B$.
We are interested in two questions:

\begin{description}
\item[(P1)] Compute $t^* = \arg\min_t \optcost(t)$ and 
$\match^* := \match_{t^*}$.
\item[(P2)] What is the combinatorial complexity of $\M(A,B)$, and how 
quickly can it be computed?
\end{description}

These questions have been studied, $p = 2$, 
by Rote~\cite{Ro10} and by 
Ben-Avraham~\etal~\cite{AHJ*}.
Two challenging, still open problems are whether the size of $\M$ is 
polynomial in both $m$ and $n$, and whether 
$t^*$ and $\match^*$ can be computed in polynomial time. 
These previous works have raised the questions only for the case 
$p = 2$, but they are open for arbitrary $p<\infty$.
There is extensive work on pattern matching and on computing similarity 
between two point sets. We refer the reader to~\cite{AG,Vel} for 
surveys. 
Here, we confine ourselves
to a brief discussion of work directly related to the problem at hand.

Much work has been done on computing a minimum-cost \emph{perfect
matching} in geometric settings. Here, $n = |A| = m = |B|=k$.
A minimum-cost perfect matching, for any $L_p$-norm, can be found in 
$\softO(n^2)$ time~\cite{AES,KMS*,PA}.\footnote{%
  The notation $\softO(\cdot)$ hides polylogarithmic factors in 
  $n$, $m$, and also polylogarithmic factors in $1/\eps$, when we only
  seek a $(1 + \eps)$-approximate solution.}
These algorithms are based on the Hungarian algorithm for 
a minimum-cost maximum matching in a bipartite graph,
and are made more efficient than the general technique by using 
certain efficient geometric data structures.
Thus, they also work when the two point sets $A$ and $B$ have different 
sizes, say, $|A| = n$ and $|B| = m$, with $k=m \leq n$. In this case, 
the running time of the algorithm is $\softO(mn)$.

Approximation algorithms for the minimum-weight perfect matching in 
geometric settings have been developed in a series of papers; see, 
e.g.,~\cite{SA122} and the references therein. For the case when the 
weight of a 
matching is the sum of the Euclidean lengths of its edges, a 
near-linear algorithm is known~\cite{SA122}. If the 
weight is the $L_p$-norm of the Euclidean lengths of the edges, for 
some $p>1$, then the best known algorithm runs in 
$\softO(n^{3/2})$ time~\cite{SA12,VA99}. In particular, for RMS matching
($p=2$) and for $p = 1, \infty$, the time for 
finding a $(1 + \eps)$-approximate 
optimal matching is $\softO(n^{3/2})$, and for a general $p$, the 
running time is 
$\softO\big(\tfrac{n^{3/2}}{\eps^{3/2}}\big)$. 
These algorithms use the scaling method by Gabow and 
Tarjan~\cite{GT89} that at each scale computes a minimum-weight matching 
by finding $n$ augmenting paths in $O(\sqrt{n})$ phases, 
where each phase takes $\softO(n)$ time (see also~\cite{GoldbergHeKaTa17}). 
If $|A| = n$, $|B|= m$, and 
$k=m \le n$, then the $m$ augmenting paths can be found in $O(\sqrt{m})$ 
phases, each of which takes $\softO(n)$ time. Hence, the total 
running time in this case is $\softO(\sqrt{m}n)$, for $p=1, 2, \infty$, 
or $\softO(\sqrt{m}n/\eps^{3/2})$, for general $p$.
When $k\le m\le n$, the minimum-weight $k$-matching is constructed, 
using the geometrically enhanced version of the Hungarian algorithm, 
in $k$ augmenting steps, each of which can be performed in 
$O(n \polylog(n))$ time. That is, the exact 
minimum-weight $k$-matching can be computed in $\softO(kn)$ time. 
The case of computing an approximate $k$-matching is somewhat trickier. 
If $k=\Theta(m)$,  one can show, adapting the technique 
in \cite{SA12}, that the running time remains 
$O(\sqrt{m}n \polylog(n))$. For smaller values of $k$, one 
can still get a bound depending on $k$, but we do not treat this case 
in the paper. It is also much less motivated from the 
point of view of applications.

Cabello~\etal~\cite{CGKR} considered
optimal shape matching under translations and/or rotations.
They considered the more general setting of \emph{weighted} point
sets, where each point of $A$ and $B$ comes with a multiplicity or
``weight''. Accordingly, the similarity criterion is
the \emph{earth-mover's distance}, or
\emph{transportation distance}, which measures the minimum
amount of work necessary to transport all the weight from $A$
to $B$, where transporting a weight $w$ by distance $\delta$ costs 
$w\cdot\delta$.
For the special case of unit weights, this reduces, via the integrality
of the minimum-cost flows,
to one-to-one matching.

We apply several ideas from Cabello~\etal's paper:
(1) the use of \emph{point-to-point translations}
to get constant-factor approximations,
(2) the selection of a random subset of these transformations to get
fast Monte Carlo algorithms, and
(3)
 tiling the vicinity of these transformations in the parameter
 space by an $\eps$-grid to get $(1+\eps)$-approximations.
 We go beyond the results of
 Cabello~\etal\ in the following aspects.
\begin{itemize}
\item 
 We give a greedy ``disk-eating'' algorithm in the space of translations
to get an improved deterministic approximation
(Theorem~\ref{disk-eating}).
 This idea could be useful for other problems.
\item
We introduce
\emph{approximate matching diagrams}: Such a diagram is a subdivision of the 
translation plane
together with a matching for each cell. This matching is approximately optimal
for every translation in the cell.
As a consequence, this diagram provides approximate
optimal matchings for \emph{all} translations.
We show that there is an
approximate matching diagram of small size, and we describe how to compute 
it efficiently (Section~\ref{approximate-diagram}).
\item
  Less importantly,
our results cover a broader class of similarity measures:
The lengths of the $k$ matching edges can be aggregated in the
objective function using any $L_p$ norm, $p\ge 1$, whereas 
Cabello~\etal only dealt with the $L_1$ norm.
By indentifying the crucial property that lies at the basis of the approximation, namely
Lipschitz continuity (Corollary~\ref{cor:shift}), this generalization comes without much additional effort.
Our results are also slightly more general because we allow outliers
(i.e., $k<\min\{m,n\}$), whereas
Cabello~\etal\ 
match the smaller set completely.
\item
By using better data structures, some of our algorithms are
more efficient.
\end{itemize}
We present approximate solutions for (P1) and (P2).
They use approximation algorithms for matching between \emph{stationary} 
sets as a black box.
We write $W(m, n, k, \eps)$ for the time that is needed to compute
a $(1+\eps)$-approximate minimum-weight matching of size   $k$
between two given (stationary) sets $A$ and $B$ of $m$ and $n$ points in
the plane,
where the weight is the $L_p$-norm of the vector or Euclidean edge
lengths, for
$k \leq \min\{m, n\}$ and for a given $\eps \geq 0$.
We abbreviate $W(m, n, k, 0)$ as simply $W(m, n , k)$. 
Table~\ref{tab:matching}
summarizes the known running times.
\begin{table}
\centering
\begin{tabular}{l | l | l | p{33mm}}
norm &  & time & reference\\
\hline
\hline
  $p \in [1, \infty]$ & exact & \raise2pt\hbox{\strut}%
                                $W(m,n,k) = \softO(kn)$ & Hungarian 
method, geo
         metric version~\cite{AES,KMS*,PA}\\
\hline
$p \in \{1, 2, \infty\}$ & $(1+\eps)$-approximate & \raise2pt\hbox{\strut}%
$W(m, n, k, \eps) = \softO(\sqrt{m}n)$ & \cite{SA12}\\
\hline
$p \in [1, \infty]$ & $(1+\eps)$-approximate 
        & \raise2pt\hbox{\strut}%
          $W(m, n, k, \eps) = \softO(\sqrt{m}n/\eps^{3/2})$ & \cite{VA99}\\
\end{tabular}

\caption{Known time bounds for various matching problems
between stationary sets. We assume $m\le n$, and in the last two
rows $k = \Theta(m)$.}
\label{tab:matching}
\end{table}
We obtain two main results:
\begin{itemize}
\item[(i)] 
We present an $\softO(mn + \tfrac{mn}{\eps^2 k} W(m, n, k, \eps/2))$-time 
algorithm for computing
  a translation vector $\tilde{t}$ and a $k$-matching $\approxmatch$ 
  between $A$ and $B$ such 
  that $\cost(\approxmatch, \tilde{t}) \leq (1 + \eps)\optcost(t^*)$.
\item[(ii)] 
 We present an 
 $\softO(mn + \tfrac{mn}{\eps^2 k} W(m, n, k, \eps/2))$-time algorithm for  
   computing a $(1 + \eps)$-approximate matching diagram of size 
   $O\big(\tfrac{n}{\eps^2}\log \tfrac{1}{\eps}\big)$, i.e.,
   a planar subdivision $\approxmap$ and a collection of
   $k$-matchings $\match_\sigma$, one matching for each face $\sigma$ 
   of $\approxmap$,
   such that for each face $\sigma$ of $\approxmap$ and for every 
   $t \in \sigma$,
   $\cost(\match_\sigma, t) \le (1+\eps) \optcost(t)$. 
\end{itemize}
The paper is organized as follows.
We start with simple solutions
to (P1) and (P2) with constant-factor 
approximations
(Section~\ref{sec:simple}).
We then refine them to obtain $(1 + \eps)$-approximate solutions,
in Section~\ref{sec:e-approx}.
Finally, 
we present improved algorithms, which attain the bounds claimed 
in (i) and (ii), in Section~\ref{sec:improved}.
All our statements hold for $p=\infty$.
In some cases, the proofs require a special treatment for this case,
but for brevity, we will mostly omit the treatment for $p=\infty$.
As in~\cite{CGKR}, the techniques used here can probably 
be extended to handle also rotations and rigid motions. We hope 
to present this extension in the full version.

\section{Simple Constant-Factor Approximations}
\label{sec:simple}

The following lemma  establishes a Lipschitz condition for the cost 
of a matching of size $k$.

\begin{lemma} \label{lem:shift}
Let $\match \subset A \times B$ be a matching of size $k$,
and let $t, \Delta \in \reals^2$ be two translation vectors.  
Then, for any $p \in [1, \infty]$, the cost under the
$L_p$-norm satisfies 
\begin{equation}\label{eq:liplp}
  \cost(\match, t + \Delta) \leq \cost(\match, t) + \|\Delta\|.
\end{equation}
\end{lemma}

\begin{proof}
  Let $\match =\{(a_1,b_1),\ldots,(a_k,b_k)\}$, and
  define two nonnegative $k$-dimensional vectors 
 $\vec{v}$ and~$\vec{w}$  by
$\vec{v}_i = \| a_i + t - b_i\|$ and 
$\vec{w}_i = \| a_i + t + \Delta - b_i \|$,
for $1\le i \le k$. By the triangle inequality for the
Euclidean norm, we have, for each $i$,
$
  \vec{w}_i =  \| a_i + t + \Delta - b_i \| \leq \| a_i + t - b_i \| + 
  \|\Delta\| = \vec{v}_i + \|\Delta\|$.
Thus, we obtain the component-wise inequality
$
  \vec{w} \leq \vec{v} + \|\Delta\| \cdot \vec{1}$,
where $\vec{1}$ denotes
the $k$-dimensional vector in which all components are $1$.
Now, 
\[
\cost(\match, t + \Delta) = \frac{\| \vec{w} \|_p}{k^{1/p}}
\leq \frac{\Big\| \vec{v} + \|\Delta\|\cdot \vec{1} 
\Big\|_p}{k^{1/p}}  \leq 
\frac{\| \vec{v} \|_p}{k^{1/p}} + \|\Delta\|\cdot 
\frac{\|\vec{1} \|_p}{k^{1/p}} =
\cost(\match, t) + \| \Delta \|,
\]
using the definition~(\ref{eq:cost}) of $\cost$, the fact that the
$L_p$-norm is a monotone function in the components whenever they are 
nonnegative,
and the triangle inequality for the $L_p$-norm.
\end{proof}
Here is an immediate corollary of Lemma~\ref{lem:shift}:
\begin{corollary}[Lipschitz continuity of the optimal cost]
\label{cor:shift}
For any two translation vectors $t_1, t_2 \in \reals^2$,
$\optcost(t_2) \le \optcost (t_1) + \|t_2 - t_1\|$. 
\end{corollary}
\begin{proof}
For the respective optimal $k$-matchings $\match_1$ and $\match_{2}$ between 
$A + t_1$ and $B$ and $A + t_2$ and $B$, 
\[ 
\optcost(t_2) =
\cost(\match_2, t_2) \leq 
\cost(\match_1, t_2) \leq 
\cost(\match_1, t_1) + \|t_2 - t_1\| 
= \optcost(t_1) + \|t_2 - t_1\|.
\qedhere
\]
\end{proof}

\subparagraph*{Approximating $t^*$ by point-to-point translations.}

As in  \cite{CGKR}, we
consider the set $T = \{ b - a \mid a \in A, b \in B\}$ of at most 
$mn$
\emph{point-to-point
translations} where some point in $A$ is moved to some point
in $B$.
The following simple observation turns out to be very useful:
\begin{lemma}[{\cite[Observation~1]{CGKR}}]
\label{lem:nn}
Let $t \in \reals^2$ be an arbitrary translation vector, and 
let $t_0\in T$ be the nearest neighbor of $t$ in $T$. Then
$\optcost(t) \ge \|t - t_0\|$.
\end{lemma}
\begin{proof}
By definition, $t_0 = b - a$ is the translation in $T$ with
$\| t - t_0 \| = \min_{(a', b') \in A \times B} \| t - b' + a' \|$.
Thus, for $p<\infty$, all summands in
the definition~\eqref{eq:cost}
of $\optcost(t)$ are at least $\| t - t_0 \|$, implying
$\optcost(t) \ge \|t - t_0\|$. The last conclusion is
trivially valid for $p=\infty$
as well.
\end{proof}

\begin{lemma}[{\cite[Lemma~1]{CGKR}}] \label{lem:apx2}
	There is a translation $t_0\in T$ with $\optcost(t_0)
        \leq 
        {2}\optcost(t^*)$.
\end{lemma}
\begin{proof}
Let $t^*$ be an optimal translation and $\match^*$ a corresponding
matching of size $k$. Take 
the translation 
$\Delta = b - a - t^*\in \reals^2$ for which $\|a + t^* - b\|$ is 
minimized, 
over $(a, b) \in \match^*$.
By Lemma~\ref{lem:nn},
$\|\Delta\| \le \optcost(t^*)$. 
The claim now follows from Lipschitz continuity 
(Corollary~\ref{cor:shift}) with
$t_1 = t^*$ and $t_2 = t^* + \Delta$, where 
the latter translation is the desired
$t_0 \in T$.
\end{proof}
We remark that for RMS matching ($p=2$), the factor~2 can
be improved to $\sqrt{2}$.
\begin{lemma}\label{lem:apxsqrt2}
	If we measure the cost under the $L_2$-norm, there is a 
	translation $t_0\in T$ with $\optcost(t_0)
        \leq \sqrt{2}\optcost(t^*)$.
\end{lemma}
\begin{proof}
We need a refined version of 
Lemma~\ref{lem:shift} for the $L_2$-norm.
For this, let $\match = \{(a_1, b_1), \dots, (a_k, b_k)\}$ be a matching 
of size $k$, and let $t, \Delta \in \reals^2$ be two translation vectors.
Define two sequences of $k$ two-dimensional vectors by
$\vec{v}_i = a_i + t - b_i$ and 
$\vec{w}_i =  a_i + t + \Delta - b_i $,
for $1\le i \le k$. Since the 
Euclidean norm is derived from a scalar product, we have, for each $i$,
\[
  \vec{w}_i^2  =  (a_i + t + \Delta - b_i )^2 = ( a_i + t - b_i )^2 + 
  2 (a_i + t - b_i) \cdot \Delta + \Delta^2 = 
  \vec{v}_i^2 +2 \vec{v}_i \cdot \Delta +  \Delta^2.
\]
Now, under the Euclidean norm, this gives
\begin{align}
\cost(\match, t + \Delta) = 
\sqrt{\frac{1}{k} \sum_{i = 1}^{k} \vec{w}_i^2}
&= \sqrt{\frac{1}{k}\sum_{i = 1}^{k} \vec{v}_i^2 +
\frac{2}{k}\sum_{i = 1}^k \vec{v}_i \cdot \Delta + \Delta^2}\notag \\
&= \sqrt{\cost(\match, t)^2 +
\frac{2}{k}\sum_{i = 1}^k \vec{v}_i \cdot \Delta + \Delta^2}. 
\label{equ:costequ}
\end{align}

Let now $t^*$ be an optimal translation and 
$\match^* = \{(a_1, b_1), \dots, (a_k, b_k)\}$ a corresponding
matching of size $k$. If $\optcost(t^*) = 0$, then we have
$t^* \in T$, and the lemma follows. Thus, assume
$\optcost(t^*) > 0$. Let 
$\vec{v}_i = a_i + t^* - b_i$ be the translated
points, for $1 \leq i \leq k$. 

Consider the vector $\gamma = \sum_{i = 1}^k \vec{v}_i$.
We claim that $\gamma = \vec{0}$.
Otherwise, by taking $\Delta = - \eps \gamma$,
for $\eps > 0$, we would get by (\ref{equ:costequ})
\[
\cost(\match^*, t^* + \Delta) = 
\sqrt{\optcost(t^*)^2 -
\frac{2}{k} \eps \gamma^2  + \eps^2\gamma^2} = 
\sqrt{\optcost(t^*)^2 
-\Big(\frac{2}{k} - \eps\Big)\eps \gamma^2},
\]
and for small enough $\eps > 0$, the translation
vector $t^* + \Delta$ would be strictly better than $t^*$,
a contradiction. Hence, for every $\Delta \in \reals^2$, we have
by (\ref{equ:costequ})
\[
\cost(\match^*, t^* + \Delta) = 
\sqrt{\optcost(t^*)^2 
+ \Delta^2}. 
\]
Now consider 
the translation 
$\Delta = b - a - t^*\in \reals^2$ for which $\|a + t^* - b\|$ is 
minimized, 
over $(a, b) \in \match^*$.
By Lemma~\ref{lem:nn},
$\|\Delta\| \le \optcost(t^*)$.  Thus, we get
\[
\cost(\match^*, t^* + \Delta) \leq 
\sqrt2{\optcost(t^2)^2}.
\]
Since $t^* + \Delta \in T$, the lemma follows.
\end{proof}

Lemma~\ref{lem:apx2} leads to the following simple algorithm for 
approximating the optimum matching. Compute $T$, and iterate over 
its elements. For each $t_0 \in T$ compute $\optcost(t_0)$ (exactly),
and return the matching with the minimum weight, 
in $O(mnW(m, n, k))$ time.

If we are willing to tolerate a slightly larger approximation factor, 
we can compute, for any $\delta > 0$ and for each $t_0 \in T$, a 
$(1 + \delta)$-approximate matching.
This approach has overall running time  $O(mnW(m, n, k, \delta))$.

\begin{theorem}
\label{const-factor-delta}
Let $A, B \subset \reals^2$, with $|A| = m$ and $|B| = n$, $m \leq n$,
and let $k \leq m$ be a size parameter. 
A translation vector 
$\tilde t \in \reals^2$ can be computed in
$O(mnW(m, n, k))$
time, such that 
$\optcost(\tilde t) \leq 2\optcost(t^*)$, where $t^*$ is the optimum 
translation.
Alternatively, for any constant $\delta > 0$, one can compute 
a translation vector $\tilde t \in \reals^2$ and a $k$-matching 
$\approxmatch$ between $A$ and $B$,
in
$O(mnW(m,n,k,\delta))$
time, such that 
$\cost(\approxmatch, \tilde t) \leq 2(1 + \delta)\optcost(t^*)$.
For the case of the Euclidean norm, this can be improved to
$\optcost(\tilde t) \leq \sqrt{2}\optcost(t^*)$ and
$\cost(\approxmatch, \tilde t) \leq \sqrt{2}(1 + \delta)\optcost(t^*)$,
respectively.
\end{theorem}

\subsection{An Approximate Matching Diagram}
\label{approximate-diagram}
We construct a planar subdivision $\approxmap$ that approximates
the matching diagram $\M$ up to factor $3$. This means
that, for each face $\sigma$ of $\approxmap$, there is a single matching
$\match_\sigma$ (that we provide) so that, for each $t \in \sigma$, we have
$
  \optcost(t) \le \cost(\match_\sigma,t) \le 3\optcost(t)$.

We need a lemma that relates the best matching for a given
translation $t$ to the closest translation in $T$.
\begin{lemma} \label{lem:apx4}
Let $t$ be an arbitrary translation, and let $t_0 \in T$ be its 
nearest neighbor in $T$, i.e., the translation
in $T$ that minimizes the length of $\Delta = t_0 - t$. Then, 
\begin{equation} \label{eq:ulb}
  \optcost(t) \le \cost(\match_{t_0}, t) \le 3\optcost(t) .
\end{equation}
\end{lemma}
(Recall that $\match_{t_0}$ denotes the optimal matching for $t_0$.)
\begin{proof}
Since $\match_{t_0}$ is a $k$-matching between $A$ and $B$, 
we have, by definition, $\optcost(t) \le \cost(\match_{t_0}, t)$. 
We prove the second inequality.
By Corollary~\ref{cor:shift},
$ 
  \optcost(t_0) \le \optcost(t) + \|\Delta\|$,
and by Lemma~\ref{lem:nn}, $\|\Delta\|\le \optcost(t)$.
Applying Lemma~\ref{lem:shift}, we obtain
\begin{align*}
\cost(\match_{t_0}, t) &\leq \cost(\match_{t_0}, t_0) + \|t - t_0\| = 
\optcost(t_0) + \|\Delta\| \\
&\le
 \optcost(t) + 2\|\Delta\| \le 
\optcost(t) + 2\optcost(t) = 3\optcost(t) .\qedhere
\end{align*}
\end{proof}

Our approximate map
$\approxmap$
is simply the Voronoi diagram $\VD(T)$,
where each cell $\VC(t_0)$, for $t_0\in T$, is associated 
with the optimal
matching $M_{t_0}$ at $t_0$.
 Correctness follows
immediately from Lemma~\ref{lem:apx4}. 
Since the complexity of $\VD(T)$ is $O(|T|)=O(mn)$, we have a 
diagram of complexity $O(mn)$. For each point $t_0 \in T$, we can 
either directly compute an optimal $k$-matching between $A + t_0$ and $B$ 
and associate the resulting map with $\VC(t_0)$, or use the 
$(1+\delta)$-approximation algorithm of \cite{SA12}. 
In the former case, $\VD(T)$ is a $3$-approximate 
matching diagram, and in the latter case it is a $3(1+\delta)$-approximate matching diagram.
We thus conclude the following:
\begin{theorem}
Let $A, B \subset \reals^2$, with $|A| = m$ and $|B| = n$, $m \leq n$, 
and let $k \leq m$ be a size parameter.
There is a $3$-approximate $k$-matching diagram of $A$ and $B$ of 
size $O(mn)$, and it (and the matchings in each cell) can be 
computed in $O(mnW(m,n,k))$ time. 
Alternatively, a $3(1+\delta)$-approximate 
matching diagram, for constant $\delta > 0$, of size $O(mn)$ can 
be computed, using the same planar decomposition, 
in $O(mnW(m,n,k,\delta))$ time.
\end{theorem}

For $p = 2$, there is an alternative, potentially better approximating, construction. For 
each $t \in T$, define the function 
$f_t (s) := \cost(\match_{t}, s)$, and 
set $F = \{ f_t \mid t \in T\}$. We let $\approxmap_0$ be the 
minimization diagram of the functions in $F$. 
Simple algebraic manipulations, similar to those for Euclidean 
Voronoi diagrams, show that $\approxmap_0$ is the minimization diagram 
of a set of $|T| \leq mn$ linear functions, namely, the functions 
$\tilde{f}_t (s) = 2\sum_{(a, b) \in \match_t} \langle a - b, s\rangle 
+ \sum_{(a, b) \in \match_t} \|a - b\|^2$, 
for $t \in T$. The resulting map $\approxmap_0$ is a 
$3$-approximate diagram of complexity $O(mn)$.
To see this, consider a Voronoi cell $\VC(t_0)$ in $\approxmap$.
We divide it into subcells in $\approxmap_0$,
each associated with some matching. All these matchings, other than 
$\match_{t_0}$, have 
smaller weights than the matching computed for $t_0$, over their 
respective subcells. Note that this subdivision is only used
for the analysis, the algorithm outputs the
original minimization diagram.
We emphasize that this construction works 
only for $p = 2$, while the Voronoi diagram applies for any 
$p \in [1, \infty]$.

For $p = 2$, using the fact that 
the Euclidean norm is derived from a scalar product, we can
improve the constant factors in Lemma~\ref{lem:apx2} and
Lemma~\ref{lem:apx4}. However, we chose to present the more
general results, since they are simpler and since we derive
a more powerful approximation below anyway.

\section{Improved Approximation Algorithms}
\label{sec:e-approx}

\subparagraph*{Computing a $(1+\eps)$-approximation of the optimum 
matching.}

This algorithm uses the same technique that was used by Cabello \etal
\cite[Section 4.1, Theorem 6]{CGKR} in a slightly different setting.
We include the description of this algorithm as a preparation for the
approximate minimization diagram, and for the improved solutions in 
the following section.

Let $t^*$ be the optimum translation, as above. Our goal is to compute 
a translation $\tilde{t}$ and a matching $\approxmatch$ so that 
$\cost(\approxmatch, \tilde t) \leq (1+\eps)\optcost(t^*)$.

Suppose we know the translation $t_0 \in T$ that minimizes the length 
of $\Delta = t_0-t^*$. By
 Lemma~\ref{lem:nn}
and Lipschitz continuity (Corollary~\ref{cor:shift}),
$\|\Delta\|
\le \optcost(t^*) \le
\optcost(t_0) \le \optcost(t^*)+\|\Delta\|
\le 2\optcost(t^*)$. 
Using
Theorem~\ref{const-factor-delta} with $\delta=1/2$,
we compute a 3-approximation
for $\optcost(t^*)$, 
This allows us to choose some radius $r_0$
with
$2\optcost(t^*)\le r_0 \le 6\optcost(t^*)$.
We take the disk $D_0$ of radius $r_0$ centered at $t_0$, and 
we tile it with the vertices of a square grid of side-length 
$\delta := \frac{\eps\sqrt{2}}{18}r_0 \leq 
\frac{\eps\sqrt{2}}3\optcost(t^*)$.
 We define $G_0$ as the set of  
 vertices of all grid cells that lie in $D_0$ or
 that overlap 
$D_0$ at least partially.
$G_0$~contains $O(r_0/\delta)^2=O(1/\eps^2)$ vertices.

We compute, by~\cite{SA12}, a $(1+\eps/2)$-approximate 
minimum-weight matching at each 
translation in $G_0$
 and return the one that achieves the 
smallest weight. Since $t^*$ has distance at most 
$\delta/\sqrt{2}$ from some grid vertex $g \in G_0$,
we have, again by
 Lipschitz continuity (Corollary~\ref{cor:shift}),
\[
\optcost(g) \le \optcost(t^*) + \frac{\delta}{\sqrt{2}} 
\le \optcost(t^*)+\frac{\eps}3\optcost(t^*) \le 
\left (1+\frac{\eps}{3}\right)\optcost(t^*).
\]
Since we compute a $(1+\eps/2)$-approximate matching for each grid 
point, the best computed matching has cost at most 
$(1+\eps/3)(1+\eps/2)\optcost(t^*)\le(1+\eps)\optcost(t^*)$,
assuming $\eps\le 1$.

Since we do not know $t_0$, we apply this procedure to all 
$mn$ translations of $T$,
for a total of
$
O( mn/\eps^2)
$ approximate matching calculations for fixed sets.

\begin{theorem}
Let $A, B \subseteq \reals^2$, $|A| = m\le |B| = n$, and
let $k\le m$ be a size parameter and $0<\eps\le 1$ a constant. 
A translation vector $\tilde t\in\reals^2$ 
and a matching $\approxmatch$ of size $k$ between $A$ and $B$ can be
computed 
in $O\big(\frac{mn}{\eps^2} \cdot W(m,n,k,\frac{\eps}{2})\big)$ time, such that 
$\cost(\approxmatch, \tilde t) \le (1+\eps)\optcost(t^*)$. 
\end{theorem}

Cabello~\etal~\cite[Theorem 4]{CGKR} give an 
$O\big(\frac{n^3m}{\eps^4}\log^2 n\big)$-time algorithm for
  the weighted problem, which includes the matching problem with
  $k=m\le n$ as a special case. It follows the same technique: it solves $O(mn/\eps^2)$ problems, each with
a fixed translation, but each such problem takes longer than in our
case because it uses the earth mover's distance.

\subparagraph*{A $(1+\eps)$-approximation of $\M$.}
We now construct a $(1+\eps)$-approximate matching diagram
$\approxmap$ of $A$ and $B$ by refining $\VD(T)$.
Without loss of generality, we assume that $\eps = 2^{-\alpha}$, for some
natural number $\alpha$, and we set
 $u := \log_2 (1/\eps) + 2 = \alpha + 2$.
We subdivide each Voronoi cell of $\VD(T)$ into smaller subcells, 
as follows. Fix $t_0 \in T$. 
For $i = 0, \dots, u$, 
let $B_i$ be the square of side-length $2^{i}\optcost(t_0)$, 
centered at $t_0$. 
Set
$B_{-1}=\emptyset$. For $i = 0, \dots, u$,
we partition $B_i\setminus B_{i-1}$ into a uniform grid
with side-length $\eps 2^{i-3}\optcost(t_0)$.
We clip each grid cell $\tau$ to $\VC(t_0)$, i.e., if 
$\tau \cap \VC(t_0) \neq \emptyset$, 
we take $\tau \cap \VC(t_0)$ as a face of $\approxmap$. 
Let $t_\tau$ be the center of the grid cell $\tau$. 
We associate $\match_\tau := \match_{t_\tau}$ with the face 
$\tau \cap \VC(t_0)$. Finally, each connected component of 
$\VC(t_0)\setminus B_u$ becomes a (possibly non-convex) face of 
$\approxmap$. There are at most four such faces, and we associate
$\match_{t_0}$ with each of them. 

\begin{figure}[ht]
  \centering
  \includegraphics[scale=0.7]{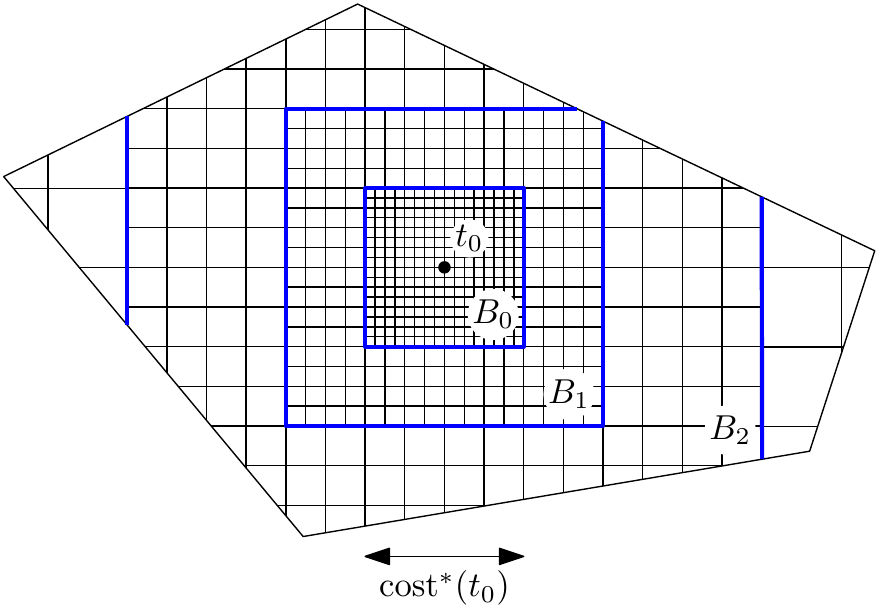}
  \caption{Partition of a Voronoi cell into nested grids, for the (unrealistically
    large) choice $\eps = 1/2$}
 \label{fig:grid}
\end{figure}

The above procedure partitions $\VC(t_0)$ into 
$O(\tfrac{1}{\eps^2}\log \tfrac1\eps)$ cells,
and their total complexity is 
$O(k_0 + \tfrac{1}{\eps^2}\log\tfrac{1}{\eps})$, where $k_0$ is the 
number of vertices on the boundary of $\VC(t_0)$. 
We repeat our procedure for all Voronoi cells of $\VD(T)$. 
Since the total complexity of $\VD(T)$ is $O(mn)$, the total complexity 
of $\approxmap$ is 
$O(\tfrac{mn}{\eps^2}\log\tfrac1\eps)$. 
\begin{lemma} \label{approxmap}
$\approxmap$ is a $(1+\eps)$-approximate matching diagram of $A$ and $B$.
\end{lemma}
\begin{proof}
Let $t \in \reals^2$ be an arbitrary translation vector,  and
let $t_0 \in T$ be the nearest neighbor of $t$ in $T$, i.e.,
$t \in \VC(t_0)$.
First, consider the case when $t\not\in B_u$. Then 
$\|t - t_0\| \ge {2\optcost(t_0)}/{\eps}$ and $\match_{t_0}$ 
is the matching associated with the 
cell of $\approxmap$ containing $t$.  Hence, using 
Lemmas~\ref{lem:shift} and~\ref{lem:nn}, we obtain
\[
\optcost(t) \le \cost(\match_{t_0}, t) \le \optcost(t_0) + \|t-t_0\| \le 
\left(1+\frac{\eps}{2}\right)\|t-t_0\| \le 
	\left (1+\frac{\eps}{2}\right)\optcost(t).
\]
Suppose $t \in B_0$. Then  
$\|t - t_0\| \le \optcost(t_0)/\sqrt{2}$. Therefore,
by Corollary~\ref{cor:shift},
\[
\optcost(t) \ge \optcost(t_0) - \|t - t_0\| \ge 
\optcost(t_0) - \frac{1}{\sqrt{2}}\optcost(t_0) = 
\left(1-\frac{1}{\sqrt{2}}\right)\optcost(t_0).
\]
Let $\tau$ be the grid cell inside $B_0$ containing $t$, and 
let $t_\tau$ be the center of $\tau$. 
Then $\|t - t_\tau\| \le \tfrac{\eps}{8\sqrt{2}}\optcost(t_0)$.  
By Corollary~\ref{cor:shift},
$ 
\optcost(t_\tau) \le \optcost(t)+\|t-t_\tau\|$. 
Furthermore,
\begin{align*}
\cost(\match_\tau, t) 
&\le \cost(\match_\tau, t_\tau) + \|t - t_\tau\| 
=\optcost (t_\tau)+\|t - t_\tau\|\\
&\le\optcost(t) + 2\|t-t_\tau\| \le 
\optcost(t) + \frac{\eps}{4\sqrt{2}}\optcost(t_0) \\
& \le \optcost(t) + 
\frac{\eps}{4\sqrt{2}}\cdot\frac{\sqrt{2}}{\sqrt{2}-1} \optcost (t) 
	\le (1+\eps)\optcost(t) .
\end{align*}
Finally, suppose $t \in B_i \setminus B_{i-1}$, for some $i \ge 1$. 
Since $t \not\in B_{i-1}$, we have $\|t-t_0\| \ge 2^{i-2}\optcost(t_0)$. 
Let $\tau$ be the grid cell of $B_i\setminus B_{i-1}$ containing $t$, 
and let $t_\tau$ be its center. Then 
$ 
\|t-t_\tau\| \le \frac{2^{i-3}}{\sqrt{2}} \eps \cdot\optcost(t_0)$. 
Starting with the inequality that was established above, we get
\begin{align*}
\cost(\match_\tau, t) 
&\le \optcost(t) + 2\|t-t_\tau\| 
\le \optcost(t) + 2\frac{2^{i-3}\eps}{\sqrt{2}}\optcost(t_0)\\
&\le \optcost(t) + \frac{\eps}{\sqrt{2}} \|t-t_0\| 
\le \optcost(t) + \frac{\eps}{\sqrt{2}}\optcost(t)
\le (1+\eps)\optcost(t).\qedhere
\end{align*}
\end{proof}

Similar to the $O(1)$-approximate matching diagram, we can improve 
the construction time by
setting $\eps'=\eps/3$ instead of $\eps$
and
 computing 
a $(1+ \eps/2)$-approximate optimal matching 
(instead of the exact matching) for the center of every cell:
\begin{theorem}
Let $A, B \subseteq \reals^2$, with $|A| = m$, $|B| = n$, $m \leq n$ and 
a size parameter $k \leq m$. 
For $0<\eps\le 1$, one can compute a $(1+\eps)$-approximate $k$-matching 
diagram of $A$ and $B$, 
of size $O(\frac{mn}{\eps^2}\log \frac{1}{\eps})$, in 
$O(\frac{mn}{\eps^2}\log \frac{1}{\eps})W(m,n,k,\frac\eps2)$ time.
\end{theorem}

\section{Improved Algorithms}
\label{sec:improved}

We now present techniques that improve, by a factor of 
$m$ or of $k$, both algorithms for computing an approximate optimal 
matching 
and an approximate matching diagram. These algorithms
work well for the case $k\approx m$, and they deteriorate when $k$ becomes small.
The first algorithm
is based on an idea
of
Cabello~\etal~\cite[Lemma~2]{CGKR}: The best matching contains a
substantial number of edges whose length does not exceed the optimum 
cost by
more than a constant factor (cf. Lemma~\ref{markov}).
This gives a randomized constant-factor approximation
algorithm that requires $O(mn/k)$ approximate
matching
computations between stationary sets in order to succeed with probability $\frac12$
(Theorem~\ref{random}).
We proceed to an improved algorithm that computes a
 constant-factor approximation with  the same number
of
fixed-translation matching calculations \emph{deterministically}.
By tiling the vicinity
of each candidate translation by an $\eps$-grid, we then obtain a
$(1+\eps)$-approximation 
(Theorem~\ref{disk-eating}).

Markov's inequality bounds the number of items in a sample that are
substantially above average. We will use the following consequence of
it:
\begin{lemma}
  \label{markov}
  Let $M$ be a matching of size $k$ between a (possibly
  translated) set $A$ and a set~$B$, with cost~$\mu$. Let $0<c\le 1$.
Then the number of pairs $(a,b)\in M$ for which
$\|a-b \| < (1+c) \mu$ is at least
$k-k/(1+c)^{p}$.
\end{lemma}
\begin{proof}
For $p=\infty$, we interpret $(1+c)^p$ as $\infty$, and
  the result is obvious 
 because
$\|a-b \| < (1+c) \mu$ for
  \emph{all} pairs $(a,b)$.
For
  $1\le p<\infty$,
    we argue by
  contradiction. The total number of pairs is~$k$. If there were more than
  $k/(1+c)^{p}$ pairs
   $(a,b)\in M$ with
   $\|a-b \| \ge (1+c) \mu$, the total cost would be
\[
\mu =
\cost(M)=
\left[\tfrac 1k \cdot \sum\nolimits_{(a,b)\in M} \|a-b \|^p\right]^{1/p} >
   \left[\tfrac 1k \cdot k/(1+c)^{p} \cdot ((1+c) \mu)^p
   \right]^{1/p}
   =
  \mu.
  \qedhere \]
\end{proof}

Consider the optimal translation $t^*$ and the corresponding optimal 
matching $\match^*$. By the lemma,
the fraction
of the pairs $(a, b) \in \match^*$ that satisfy
$\|a + t^*- b\| \le (1+c) \optcost(t^*)$
is
at least
 $
1-1/(1+c)^{p}\ge
1-1/(e^{c/2})^{p}
= 
1-e^{-cp/2}$, since $c \leq 1$.
Hence, with probability 
at least $(1-e^{-cp/2})\frac km$,
a randomly chosen $a \in A$ will participate
in such 
a ``close'' pair of~$\match^*$. We do not know the $b \in B$ with $(a, b) \in \match^*$,
so we try all $n$ possibilities. 
That is, we choose a single random point $a_0\in A$, and we try all 
$n$ translations $b - a_0 \in T$, 
returning the minimum-weight partial matching over these translations. 
With probability at least $(1-e^{-cp/2})\frac km$,
we get, by Lemma~\ref{lem:apx4}, a matching whose weight is at most
$\optcost(t^*) + (1+c) \optcost(t^*) = (2+c)\optcost(t^*)$.
The runtime of this procedure is $n \cdot W(m,n,k)$,
or $n\cdot W(m,n,k,\delta)$
if we compute at each of the above translations~$t_0$ a $(1+\delta)$-approximation to $\optcost(t_0)$.
To boost the success probability,
we repeat this 
drawing process $s$ times
and obtain a $(2+c)(1+\delta)$-approximation to the 
best matching, with probability at least $1-\left(1-(1-e^{-cp/2})\frac km\right)^s$.
By setting
$c=\delta=\eps/4$,
 we get
the following theorem.
\begin{theorem}
  \label{random}
Let $A, B \subset \reals^2$ with $|A| = m$ and $|B| = n$,  $m \leq n$,
and let $k \leq m$ and
$s\ge 1$ be parameters. 
Then, a translation vector $\tilde t\in\reals^2$ and a
matching $\approxmatch$
of size $k$ between $A$ and $B$ can be computed
in $O(sn\cdot W(m,n,k,\eps/4))$ time, such that 
$\cost(\approxmatch,\tilde t) \le (2+\eps)\optcost(t^*)$ with probability
at least $1 - \left(1-(1-e^{-\eps p/8})\frac{k}{m}\right)^s
$, for any $\eps$ with $0<\eps\le1$.
\qed
\end{theorem}
If $\eps p$ is small, the probability is approximately equal to the
simpler
expression $1-e^{-s\cdot \eps pk/8m}$.

Cabello et~al.~\cite{CGKR}
proceeded from this result to
a $(1+\eps)$-approximation
by
tiling the vicinity of each selected translation with an
$\eps$-grid~\cite[Theorem~7]{CGKR}. We will first replace the
randomized algorithm by a deterministic one, and apply the
$\eps$-grid refinement afterwards.

We now describe a deterministic algorithm for approximating $t^*$ 
and the corresponding matching $\match^*$. 
At a high level, the $mn$ points of $T$ are partitioned into 
$O(mn/k)$ clusters of size $\Omega(k)$, and one point, not necessarily 
from $T$, is chosen to represent each cluster.
We will argue that the point in the resulting set $\X$ of representatives
that is nearest to $t^*$ yields a matching whose value
at $t^*$ is an $O(1)$-approximation of $\optcost(t^*)$.

Here is the main idea of how we cluster the points in $T$ and 
construct $\X$, in an incremental manner. In step $i$, we greedily choose 
the smallest disk $D_i$ that contains $k/2$ points of $T$ 
(or all of $T$, if $|T| \le k/2$), add the center of $D_i$ to $\X$, 
delete the points of $D_i\cap T$ from $T$, and repeat. 
Carmi~\etal~\cite{CDM*} have described 
an efficient algorithm for this clustering problem. It 
preprocesses $T$ into a data structure (consisting of three compressed 
quadtrees) in  $O(mn\log n)$ time, so that
in step $i$, the disk $D_i$ can be computed in $\softO(k^2)$ time
and 
$D_i\cap T$ can be deleted from the data structure in $\softO(k^2)$ 
time, leading to an $\softO(mnk)$-time 
algorithm. They also present a faster approximation algorithm for this 
clustering problem: in step $i$, 
instead of computing the smallest enclosing disk $D_i$, they show 
that a disk of radius at most twice that of $D_i$ that still contains $k/2$ 
points of $T$ can be computed in $\softO(k)$ time, and that 
$D_i\cap T$ can be deleted in $\softO(k)$ time, thereby improving the 
overall running time to $\softO(mn)$. This approximation algorithm is 
sufficient for our purpose. We next give a more formal description of
our method:

At the beginning of step $i$, we have a set $P_i \subseteq T$ and 
the current set $\X$.  Initially,
$P_1 = T$ and  $\X = \emptyset$.  We preprocess $P_1$, in 
$\softO(|T|)=\softO(mn)$ time, into the data structure described by 
Carmi~\etal~\cite{CDM*}.
We perform the following operations in step $i$: if 
$P_i=\emptyset$, the algorithm terminates.
If $0<|P_i| \le k/2$, we compute the smallest disk $D_i$ containing
$P_i$. If $|P_i| > k/2$, then let $\rho_i^*$ be the radius of the
smallest disk that contains at least $k/2$ points of $P_i$. Using
the algorithm in~\cite{CDM*}, we compute a disk $D_i$ of radius 
$\rho_i\le 2\rho_i^*$ containing at least $k/2$ points of $P_i$. 
We add the center $\xi_i$ of $D_i$ to $\X$, and we set 
$P_{i+1}:=P_i \setminus D_i$. We remove $P_i\cap D_i$ from the 
data structure, as described in~\cite{CDM*}.
Let $\D$ be the set of disks computed by the above procedure. By 
construction, $\rho_i^* \le \rho_{i+1}^*$, $\rho_i \le 
2\rho^*_i \le
2\rho^*_{i+1}\le
2\rho_{i+1}$,
and $|\X|=|\D|\le 2mn/k$. 
The following two lemmas establish the correctness of our method. 
\begin{lemma}
\label{lem:cluster}
Let $t\in\reals^2$ be a translation vector, and let $\xi_0$ be its 
nearest neighbor in $\X$.  Then $\|t-\xi_0\| \le 
3\cdot 2^{1/p}\optcost(t)$.
\end{lemma}
\begin{proof}
Let $D$ be the disk of radius $2^{1/p}\optcost(t)$ centered at $t$, 
and let $S=D\cap T$. By
 Lemma~\ref{markov} with $1+c = 2^{1/p}$,
we have $|S| \ge k/2$.
Let $D_i$ be the first 
disk chosen by the above procedure that contains 
a point $t_0$ of $S$, so $S \subseteq P_i$. We must have 
$\rho_i^* \le 2^{1/p}\optcost(t)$, because the smallest disk 
that contains at least $k/2$ points of $P_i$ is not larger than $D$. 
Hence, $\rho_i \le 2\cdot 2^{1/p}\optcost(t)$, and
\begin{align*}
\|t-\xi_i\| \le \|t-t_0\| + \|t_0-\xi_i\| &\le 2^{1/p}\optcost(t)+\rho_i 
\\&\le 
2^{1/p}\optcost(t)+ 2 \cdot 2^{1/p}\optcost(t) = 3 \cdot 
2^{1/p}\optcost(t).\qedhere
\end{align*}
\end{proof}
\begin{lemma}
$\displaystyle 
\min_{\xi\in \X} \optcost(\xi) \le (1+3 \cdot 2^{1/p})\optcost(t^*)$.
\end{lemma}
\begin{proof}
Let $\xi_0$ be the nearest neighbor to $t^*$ in $\X$. Applying 
Lemma~\ref{lem:cluster} with $t=t^*$, we obtain
$\|t^*-\xi_0\|\le 3 \cdot 2^{1/p}\optcost(t^*)$. By 
Corollary~\ref{cor:shift}, we then have
$ \optcost(\xi_0) \le \optcost(t^*) + \|t^*-\xi_0\| \le
(1+3 \cdot 2^{1/p})\optcost(t^*)$.
\end{proof}

We fix a constant $\delta\in (0,1]$. We compute a 
$(1+\delta)$-approximate $k$-matching $M_\xi$ between $A+\xi$ and 
$B$, for every $\xi\in\X$, and choose the best among them. 
This will give an $O(1)$-approximation of the minimum-cost 
$k$-matching under translation.
We can extend this algorithm to yield a 
$(1+\eps)$-approximation algorithm following the 
same procedure as in Section~\ref{sec:e-approx}:
We draw a 
disk of radius $(1+3\cdot 2^{1/p}+4\eps)\optcost(t^*)$ 
around each point of $\X$.
We draw a uniform grid of cell size
$O(\eps)$ and
look at all vertices $t$ of grid cells that overlap 
one of these disks at least partially.
We
compute a $(1+\eps/2)$-approximation 
 for the best matching of size~$k$ between $A+t$ and $B$ for each 
 of the grid point $t$ under consideration, and 
we choose the best matching among them. 
 Putting everything together, 
we obtain the following:

\begin{theorem}
  \label{disk-eating}
Let $A, B \subset \reals^2$, with $|A| = m$ and $|B| = n$, and let 
$0<\eps\le 1$ and $k \le \min\{m,n\}
$ be parameters. 
Then, a translation vector $\tilde t\in\reals^2$ and 
a matching $\approxmatch$ of size~$k$ between $A$ and $B$ can be
computed
in $\softO(mn + \tfrac{mn}{\eps^2k}W(m,n,k,\frac\eps2))$ time, such that 
$\cost(\approxmatch,\tilde t) \le (1+\eps)\optcost(t^*)$.
\qed
\end{theorem}

We show that $\VD(\X)$ is indeed an $O(1)$-approximate matching diagram 
of $A$ and $B$. This is analogous to
Section~\ref{approximate-diagram}
(Lemma~\ref{lem:apx4}).
\begin{lemma}
Let $t\in\reals^2$ be a translation vector, and let $\xi_0$ be its 
nearest neighbor in $\X$. Then,
$
\optcost(t) \le \cost(\match_{\xi_0}, t) \le (1+6\cdot2^{1/p})\optcost(t)$.
\end{lemma}
\begin{proof}
Since $\match_{\xi_0}$ is a matching of size $k$ between $A$ and $B$, we have, 
by definition, 
$\optcost(t) \le \cost(\match_{\xi_0}, t)$. We now prove the 
second inequality. By Corollary~\ref{cor:shift}, 
$ 
\optcost(\xi_0) \le \optcost(t) + \|t-\xi_0\|$,
Lemma~\ref{lem:shift}, and Lemma~\ref{lem:cluster},
\begin{align*}
\cost(M_{\xi_0}, t) &\le 
\cost(M_{\xi_0}, \xi_0) + \|t-\xi_0\|
\\&
 = \optcost(\xi_0)+ \|t-\xi_0\|
\le \optcost(t) + 2\|t-\xi_0\| \le (1+6 \cdot 2^{1/p})\optcost(t).
\qedhere
\end{align*}
\end{proof}

The combinatorial complexity of $\VD(\X)$ is $O(mn/k)$. We can 
now construct a $(1+\eps)$-approximate 
matching diagram by refining each Voronoi cell of $\VD(\X)$, as 
in Section~\ref{sec:e-approx}, 
but the constants have to be chosen differently.
The diagram has
$O(\tfrac{mn}{k\eps^2}\log \tfrac{1}{\eps})$ cells, and we need
$W(m,n,k,\frac\eps2)$ time per cell.
We obtain the following:

\begin{theorem}
  \label{diagram-size}
Let $A, B \subset \reals^2$, $|A| = m \le |B| = n$, and let
$k \leq m$, $\eps \in (0, 1]$ be parameters. There exists 
a $(1+\eps)$-approximate $k$-matching diagram of $A$ and $B$ of size 
$O(\tfrac{mn}{k\eps^2}\log \tfrac{1}{\eps})$, 
and it can be computed in 
$\softO(mn) + O{\left(\tfrac{mn}{k\eps^2}\log \frac{1}\eps
    W(m,n,k,\frac\eps2)
\right)}$ time.
\qed
\end{theorem}

For the case when $cm \le k \le (1-c)n$ for some constant $c>0$,
we can show that the bound in Theorem~\ref{diagram-size} on the size of the
diagram is tight in the worst case in terms of $m$, $n$, and $k$ (but
not of~$\eps$):
 If $A$ is a unit grid of size $\sqrt{m}\times\sqrt{m}$ 
 and $B$ is a unit grid of size $\sqrt{n}\times\sqrt{n}$, then
there are $\Omega(n)$ translation vectors at which $A$ and $B$ are 
perfectly aligned and have at least $k$ points in common. Thus, any $O(1)$-approximate matching diagram of $A$
and $B$ needs to have $\Omega(n)$ distinct faces.

\bibliography{match}

\newcommand{\SortNoop}[1]{}
\begin{thebibliography}{10}

\bibitem{AES}
Pankaj~K. Agarwal, Alon Efrat, and Micha Sharir.
\newblock Vertical decomposition of shallow levels in 3-dimensional
  arrangements and its applications.
\newblock {\em SIAM J. Comput.}, 29(3):912--953, 1999.

\bibitem{AG}
Helmut Alt and Leonidas~J. Guibas.
\newblock Discrete geometric shapes: Matching, interpolation, and
  approximation.
\newblock In J.R. Sack and J.~Urrutia, editors, {\em Handbook of Comput.
  Geom.}, pages 121--153. Elsevier, Amsterdam, 1999.

\bibitem{AHJ*}
Rinat Ben-Avraham, Matthias Henze, Rafel Jaume, Bal{\'a}zs Keszegh, Orit~E.
  Raz, Micha Sharir, and Igor Tubis.
\newblock Partial-matching {RMS} distance under translation: Combinatorics and
  algorithms.
\newblock {\em Algorithmica}, 80(8):2400--2421, 2018.

\bibitem{CGKR}
Sergio Cabello, Panos Giannopoulos, Christian Knauer, and G{\"{u}}nter Rote.
\newblock Matching point sets with respect to the {E}arth {M}over's {D}istance.
\newblock {\em Comput. Geom.}, 39(2):118--133, 2008.

\bibitem{CDM*}
Paz Carmi, Shlomi Dolev, Sariel Har-Peled, Matthew~J. Katz, and Michael Segal.
\newblock Geographic quorum system approximations.
\newblock {\em Algorithmica}, 41(4):233--244, 2005.

\bibitem{GT89}
Harold~N. Gabow and Robert~Endre Tarjan.
\newblock Faster scaling algorithms for network problems.
\newblock {\em SIAM J. Comput.}, 18(5):1013--1036, 1989.

\bibitem{GoldbergHeKaTa17}
Andrew~V. Goldberg, Sagi Hed, Haim Kaplan, and Robert~E. Tarjan.
\newblock Minimum-cost flows in unit-capacity networks.
\newblock {\em Theory Comput. Syst.}, 61(4):987--1010, 2017.

\bibitem{HJK}
Matthias Henze, Rafel Jaume, and Bal{\'{a}}zs Keszegh.
\newblock On the complexity of the partial least-squares matching {V}oronoi
  diagram.
\newblock In {\em Proc. 29th European Workshop Comput. Geom. (EWCG)}, pages
  193--196, 2013.

\bibitem{KMS*}
Haim Kaplan, Wolfgang Mulzer, Liam Roditty, Paul Seiferth, and Micha Sharir.
\newblock Dynamic planar {V}oronoi diagrams for general distance functions and
  their algorithmic applications.
\newblock In {\em Proc. 28th Annu. ACM-SIAM Sympos. Discrete Algorithms
  (SODA)}, pages 2495--2504, 2017.

\bibitem{PA}
Jeff~M. Phillips and Pankaj~K. Agarwal.
\newblock On bipartite matching under the {RMS} distance.
\newblock In {\em Proc. 18th Canad. Conf. Comput. Geom. (CCCG)}, pages
  143--146, 2006.

\bibitem{Ro10}
G{\"u}nter Rote.
\newblock Partial least-squares point matching under translations.
\newblock In {\em Proc. 26th European Workshop Comput. Geom. (EWCG)}, pages
  249--251, 2010.

\bibitem{SA12}
R.~Sharathkumar and Pankaj~K. Agarwal.
\newblock Algorithms for the transportation problem in geometric settings.
\newblock In {\em Proc. 23rd Annu. ACM-SIAM Sympos. Discrete Algorithms
  (SODA)}, pages 306--317, 2012.

\bibitem{SA122}
R.~Sharathkumar and Pankaj~K. Agarwal.
\newblock A near-linear time $\varepsilon$-approximation algorithm for
  geometric bipartite matching.
\newblock In {\em Proc. 44th Annu. ACM Sympos. Theory Comput. (STOC)}, pages
  385--394, 2012.

\bibitem{VA99}
Kasturi~R. Varadarajan and Pankaj~K. Agarwal.
\newblock Approximation algorithms for bipartite and non-bipartite matching in
  the plane.
\newblock In {\em Proc. 10th Annu. ACM-SIAM Sympos. Discrete Algorithms
  (SODA)}, pages 805--814, 1999.

\bibitem{Vel}
Remco~C. Veltkamp.
\newblock Shape matching: Similarity measures and algorithms.
\newblock In {\em Proc. Intl. Conf. Shape Modeling and Applications}, pages
  188--197, 2001.

\end{thebibliography}

\end{document}